\documentclass[a4paper, 11pt]{article}
\usepackage{amsthm}
\usepackage{amsfonts}
\usepackage[T1]{fontenc}

\newtheorem{theorem}{Theorem}[section]
\newtheorem{proposition}[theorem]{Proposition}

\begin{document}
\title{Noncommutative Ricci flow in a matrix geometry}
\author{Rocco Duvenhage\\
Department of Physics, University of Pretoria, Pretoria 0002 
\\South Africa
\\rocco.duvenhage@up.ac.za}
\date{2013-12-3}
\maketitle

\begin{abstract}
We study noncommutative Ricci flow in a finite-dimensional representation of a
noncommutative torus. It is shown that the flow exists and converges to the
flat metric. We also consider the evolution of entropy and a definition of
scalar curvature in terms of the Ricci flow.

\end{abstract}

\section{Introduction}

\ Recently three different approaches to Ricci flow in noncommutative geometry
were initiated in \cite{V}, \cite{BM} and \cite{CM} respectively, the latter
two focussing on the case of a noncommutative torus. In this paper we develop
a more elementary approach in the case of a simple matrix geometry, namely a
finite-dimensional representation of the rational noncommutative torus often 
called the fuzzy torus.

The setting we use, in particular the type of the metric we consider, is of
the same form as that of \cite{CT} which also forms the basis for \cite{BM}
and \cite{CM}, although due to our finite-dimensional setup the situation is
simpler and we do not make direct use of results in \cite{CT}. Our formulation
of the Ricci flow is more direct than the spectral approach of \cite{BM},
and similar to that of \cite{CM}.

We note that since the introduction of Ricci flow by Hamilton \cite{H} as an
approach to the Thurston geometrization conjecture and in particular the
Poincar\'{e} conjecture, it has been very useful in geometry and topology, as
is by now well-known, in particular due to the work of Perelman \cite{P1},
\cite{P2}, \cite{P3} proving the Poincar\'{e} conjecture. It should also be
mentioned though, that at about the same time that Hamilton introduced Ricci
flow, it appeared independently and for very different reasons in the study
of renormalization of sigma models in physics by Friedan \cite{F}, \cite{F2}.
Considering the importance of Ricci flow in geometry, topology and physics on
the one hand, and the development and potential physical applications of
noncommutative geometry on the other, it is natural to
explore Ricci flow in the noncommutative framework. This is the general
motivation for this paper.

More specifically we are motivated by the goal to have a simple version of
noncommutative Ricci flow. The approach followed here is concrete enough that
one can prove certain results that in the other approaches appear difficult,
for example we show that in our setting the metric flows to the flat metric as
it does in the case of a classical torus \cite{H2}. It should be kept in mind
though that this does not solve the corresponding, but more difficult problem
for the noncommutative torus mentioned in \cite{BM}. Part of the reason that
we can prove results like these, is that in this approach in terms of a 
finite-dimensional representation, the Ricci flow is given by a system of first order
ordinary differential equations (rather than a partial differential equation
as in the classical case).

The latter point will become clear in Section 2, where we also briefly review
the ideas from matrix geometry that we need, and describe the metrics that we
consider. In Section 3 the existence and convergence to the flat metric of the
Ricci flow is shown. In Section 4 we explain how the metric can be viewed as a
density matrix and study the evolution of the resulting von Neumann entropy
under the Ricci flow. Section 5 briefly considers how the scalar curvature of
the matrix geometry can be unambiguously defined via the Ricci flow, and is
closely related to the work in \cite{CM}. The relevance of this is that, as
discussed for example in \cite{CM}, \cite{FK} and \cite{A}, defining scalar
curvature is somewhat more subtle in noncommutative geometry than in the
classical case. Noncommutative Ricci flow provides a particularly elegant
approach to defining scalar curvature, at least in our framework. Finally, in
Section 6, we outline possible further work that can be explored.

\section{Basic ideas and definitions}

Here we review relevant background and set up the framework used in the rest
of the paper.

Recall that Ricci flow is given by \cite{H}
\[
\frac{\partial g_{\mu\nu}}{\partial t}=-2R_{\mu\nu}
\]
where $g_{\mu\nu}$ is a metric on some differentiable manifold, $R_{\mu\nu}$
is the corresponding Ricci tensor, and $t$ is a real variable (``time''). If
we restrict ourselves to surfaces and to metrics of the form $g_{\mu\nu
}=c\delta_{\mu\nu}$ where $\delta_{\mu\nu}$ is the Kronecker delta and $c$ is
some strictly positive function on the surface (a conformal rescaling factor),
then it can be shown that the above Ricci flow equation reduces to
\[
\frac{\partial c}{\partial t}=\partial_{\mu}\partial_{\mu}\log c
\]
where $\partial_{\mu}$ is the partial derivative with respect to the
coordinate $x^{\mu}$ on the surface, and we sum over the repeated index
$\mu=1,2$. An interesting metric of this form is the cigar metric of Hamilton
\cite{H2}, also known as the (euclidean version of) the Witten black hole
\cite{W}, \cite{MSW}, \cite{Ve}, given by
\[
c=\frac{1}{M+x^{\mu}x^{\mu}}
\]
where the parameter $M>0$ is independent of $x^{\mu}$ and can be interpreted
as describing the mass of the black hole. In the case of the manifold
$\mathbb{R}^{2}$, this metric is a soliton for the Ricci flow, however as one
might expect from the general theory of Ricci flow for compact surfaces (see
in particular \cite[Theorem 10.1]{H2}) solitons other than the flat metric do not
appear in the noncommutative case below where we consider a noncommutative
version of the torus, which classically is a compact manifold.

Next we recall the matrix geometry we are going to work with. The basic ideas
regarding matrix geometries originated in \cite{Ho} and \cite{M}, the latter
including a brief discussion of fuzzy tori. Also see for example 
\cite[Section 3.1]{M2}, \cite{FFZ} and \cite{Ho2} for work more directly
related to fuzzy tori. We consider two unitary $n\times n$ matrices $u$ and $v$ 
satisfying the commutation relation
\[
vu=quv
\]
where
\[
q=e^{2\pi im/n}
\]
for an $m\in\left\{1,2,...,n-1\right\}$ such that $m$ and $n$ are relatively
prime; note in particular that $q^{n}=1$, but $q^{j}\neq1$ for $j=1,...,n-1$. 
This commutation relation is similar to the Weyl form of the commutation 
relation between position and momentum of a quantum particle in one dimension. 
Indeed, the analogy is very close, since as in the case of position and momentum, 
$u$ and $v$ are connected by a Fourier transform, but for the cyclic group 
$\mathbb{Z}_{n}=\left\{0,1,2,...,n-1\right\}$ instead of the group $\mathbb{R}$ 
as for position and momentum. It is worth mentioning that matrices of this form 
already appeared long ago in the physics literature, for example in \cite{S}, 
\cite{tH} and \cite{B}. Concretely we can use
\[
u=\left[
\begin{array}
[c]{cccc}
1 &  &  & \\
& q &  & \\
&  & \ddots & \\
&  &  & q^{n-1}
\end{array}
\right]
\qquad\mbox{and}\qquad
v=\left[
\begin{array}
[c]{cccc}
0 & 1 &  & \\
& 0 & \ddots & \\
&  & \ddots & 1\\
1 &  &  & 0
\end{array}
\right]
\]
where the blanks spaces are filled with zeroes. It is then an easy exercise to
verify that as mentioned above, $v$ is obtained by the Fourier transform of
$\mathbb{Z}_{n}$ from $u$, i.e. $v=F^{\ast}u\,F$ where $F$ is the Fourier
transform on $\mathbb{Z}_{n}$ in a suitable basis, namely the $n\times n$ 
unitary matrix $F_{jk}=q^{-jk}/n^{1/2}$ for $j,k=0,1,...,n-1$, and $F^{\ast}$ 
denotes its Hermitian adjoint. It is also easily
checked that $u$ and $v$ generate the algebra $M_{n}$ of all $n\times n$
complex matrices, namely the commutant of $\left\{u,v\right\}$ consists of
scalar multiples of the identity matrix, $\left\{u,v\right\}^{\prime
}=\mathbb{C}1$, where here and later we denote the $n\times n$ identity matrix
simply as $1$, since no confusion will arise.

We now take any two Hermitian matrices $x$ and $y$ such that
\[
u=e^{\frac{2\pi i}{n}x}\quad\mbox{and}\quad v=e^{\frac{2\pi i}{n}y}
\]
and define derivations $\delta_{1}$ and $\delta_{2}$ on the algebra $M_{n}$ by
the commutators
\[
\delta_{1}:=\left[y,\cdot\right]\quad\mbox{and}\quad\delta_{2}:=-\left[x,\cdot\right]
\]
which are analogues of the derivatives $\frac{1}{i}\partial_{\mu}$ in the
classical case above. Note that we have made a specific choice of real
multiplicative constants in these derivations, and moreover $x$ and $y$ are
not uniquely determined, but our analysis will not depend on the exact choices
made. 

To gain some insight about these derivations, we calculate $\delta_{1}u$ and
$\delta_{2}v$, though only for a specific choice of $x$ and $y$. Note that in
general from the definitions of $u$ and $F$ above, we have $FuF^{\ast}=
\left(  F^{\ast T}u^{T}F^{T}\right)  ^{T}=
\left(  F^{\ast}uF\right)  ^{T}=v^{T}$ where $a^{T}$ 
denotes the transpose of a matrix $a$. Also note that $x$ is a diagonal
matrix (with real values on its diagonal). As part of our specific choice we
choose $y:=F^{\ast}xF$, given a choice of $x$. It then follows that
\[
\delta_{1}u=F^{\ast}\left[  x,FuF^{\ast}\right]  F=
F^{\ast}\left[x^{T},v^{T}\right]  F=
F^{\ast}\left[  v,x\right]  ^{T}F=
F^{\ast}\left(  \delta_{2}v\right)  ^{T}F
\]
If we now consider the specific (and most obvious) choice
\[
x=\left[
\begin{array}
[c]{ccccc}
0 &  &  &  & \\
& m &  &  & \\
&  & 2m &  & \\
&  &  & \ddots & \\
&  &  &  & (n-1)m
\end{array}
\right]
\]
it is easily verified that
\[
\delta_{2}v=mv-mne_{n1}
\]
where $e_{jk}$ denotes the $n\times n$ matrix with $1$ in the $(j,k)$-position 
and zeroes elsewhere. From the formulas above it then
follows that
\[
\delta_{1}u=mu-mnF^{\ast}e_{1n}F
\]
These formulas are analogous to the classical formulas 
$\frac{1}{i}\partial_{2}e^{imx^{2}}=me^{imx^{2}}$ and 
$\frac{1}{i}\partial_{1}e^{imx^{1}}=me^{imx^{1}}$ 
respectively in terms of classical real coordinates $x^{\mu}$, but with terms 
included to compensate for finite-dimensionality. Similarly one can consider
other choices of $x$. For example taking all the entries of our current choice
modulo $n$, leads to similar formulas for $\delta_{1}u$ and
$\delta_{2}v$, but with different ``compensation'' terms which are notationally
a bit more difficult to write down than in the formulas for $\delta_{1}u$ and
$\delta_{2}v$ above. Of course, for general choices of $x$ and $y$ we
obviously also have $\delta_{1}v=0$ and $\delta_{2}u=0$ in perfect analogy to
the classical case.

Returning to general $x$ and $y$, from these derivations we can define a 
noncommutative analogue of a Laplacian as an operator on $M_{n}$:
\[
\triangle:=\delta_{1}^{2}+\delta_{2}^{2}=\delta_{\mu}\delta_{\mu}
\]
The importance of this for our purposes is clear from the classical Ricci flow
equation above where $\partial_{\mu}\partial_{\mu}$ corresponds to
$-\delta_{\mu}\delta_{\mu}$. We are using conventions ensuring that
$\triangle$ is a positive operator (see Proposition 2.1 below) and that's why
$\triangle$ corresponds to $-\partial_{\mu}\partial_{\mu}$ rather than
$\partial_{\mu}\partial_{\mu}$.

Regarding notation, when we write $\delta_{\mu}\delta_{\mu}$, we are summing
over $\mu$, but when we write $\delta_{\mu}^{2}$ there is no sum over $\mu$.

Keep in mind that $M_{n}$ is an involutive algebra, i.e. a $\ast$-algebra,
with the involution given by the usual Hermitian adjoint of a matrix $a$, which 
we denote by $a^{\ast}$ to fit into the standard C*-algebraic notation (in
conventional quantum mechanical notation it is denoted by $a^{\dagger}$).
Since the algebra is finite-dimensional, all norms on it are equivalent (so
they give the same topology on $M_{n}$) and complete. In the operator norm 
$M_{n}$ is indeed a unital C*-algebra, however the theory of C*-algebras will 
not be needed in this paper, though it is conceptually useful to think in
terms of C*-algebras, since a unital C*-algebra is a
noncommutative analogue of the C*-algebra of continuous complex-valued
functions on a compact topological space. The equivalence of norms will be 
useful in a technical sense, since it will allow us to use whichever norm is 
most convenient in any given situation. In particular the Hilbert-Schmidt norm 
obtained from the inner product
\[
\left\langle a,b\right\rangle :=\tau\left(a^{\ast}b\right)
\]
on $M_{n}$ will come into play in the next section. Here $\tau$ denotes the
usual trace on $M_{n}$, i.e. the sum of the diagonal elements of a matrix. We
think of $\tau\left(  a\right)  $ as a noncommutative integral of the
complex-valued ``function'' $a$, corresponding in the classical case to the
integral with respect to Haar measure on the torus.

We have the following simple proposition which mimics properties of partial
derivatives of complex-valued functions of two real variables on the classical
torus. For example, the first property corresponds to the fact that if both the
first partial derivatives of a function are zero globally, then the function
is constant.

\begin{proposition}
The following properties hold (where $a,b\in M_{n}$ are arbitrary):

(a) If $\delta_{1}a=\delta_{2}a=0$, then $a\in\mathbb{C}1$. Conversely,
$\delta_{\mu}1=0$.

(b) We can integrate by parts, i.e. $\tau\left(a\delta_{\mu}b\right)
=-\tau\left(b\delta_{\mu}a\right)$.

(c) The derivations $\delta_{1}$ and $\delta_{2}$ are Hermitian, i.e.
$\left\langle a,\delta_{\mu}b\right\rangle =\left\langle\delta_{\mu}a,b\right\rangle$, 
and furthermore $\left(\delta_{\mu}a\right)^{\ast}=-\delta_{\mu}\left(a^{\ast}\right)$.

(d) The operators $\delta_{1}^{2}$, $\delta_{2}^{2}$ and $\triangle$ on the
Hilbert space $M_{n}$ are positive, i.e. 
$\left\langle a,\delta_{\mu}^{2}a\right\rangle\geq0$ and $\left\langle a,\triangle a\right\rangle\geq0$.

(e) If $\left\langle a,\delta_{\mu}^{2}a\right\rangle =0$, then $\delta_{\mu}a=0$, 
for each value of $\mu$ separately.

(f) $\ker\triangle=\mathbb{C}1$

(g) $\tau\left(\triangle a\right)=0$
\end{proposition}

\begin{proof}
(a) Clearly if a matrix commutes with $x$ then it commutes with $u$, and if it
commutes with $y$ it commutes with $v$, so if it commutes with both $x$ and
$y$ it is in $\left\{u,v\right\}^{\prime}=\mathbb{C}1$. The converse is trivial.

(b) This is verified easily from the definition of the derivations and the
fact that $\tau(ab)=\tau(ba)$.

(c) The first part is just a different way of stating (b) using the second
part, which in turn follows from the definition of the derivations.

(d) and (e) follow from (c) and the fact that $\tau$ is faithful, i.e.
$\tau\left(a^{\ast}a\right)=0$ implies that $a=0$.

(f) Clearly $\triangle1=0$. On the other hand, if $\triangle a=0$, it follows that 
$0=\left\langle a,\triangle a\right\rangle =\left\langle a,\delta_{1}
^{2}a\right\rangle +\left\langle a,\delta_{2}^{2}a\right\rangle $ which means
that $\left\langle a,\delta_{\mu}^{2}a\right\rangle =0$ by (d) so 
$\delta_{\mu}a=0$ by (e). By (a) it then follows that $a\in\mathbb{C}1$.

(g) This follows directly from the fact that $\triangle$ is defined in terms
of commutators while $\tau$ is the trace.
\end{proof}

As will be seen, we really just need the abstract properties of the
derivations listed in this proposition in the rest of the paper, rather than
the explicit definitions of $u$ and $v$, and $\delta_{1}$ and $\delta_{2}$,
given above. Because of this and for notational convenience, we will from now
on work with the more abstract notation
\[
A=M_{n}
\]
Indeed, even in the proof of the proposition the explicit definitions of $u$
and $v$ are not used directly, though we did use the fact that they generate
the algebra, and their relation to $x$ and $y$. The finite-dimensionality of
$A$ will be used as well in the rest of the paper.

We now have an algebra $A$ with derivations, so we have the basic elements of a
noncommutative geometry. Next we need to discuss a class of metrics. Following
the idea in \cite{CT}, we describe a metric by a strictly positive element $c$
of the C*-algebra $A$, i.e. a Hermitian element of $A$ whose eigenvalues are
strictly positive. (We write $a>0$ to indicate that a Hermitian $a\in A$ is
strictly positive, and $a\geq0$ that it is positive, meaning that its
eigenvalues are larger or equal to zero.) This $c$ of course is the
noncommutative counterpart of the $c$ in $g_{\mu\nu}=c\delta_{\mu\nu}$ in the
classical case discussed above. Since we only consider metrics of this form,
we refer to such a $c$ itself simply as a \emph{noncommutative metric}. Note
that if $c$ is a scalar multiple of the identity matrix, then it can be
interpreted as a \emph{flat} metric, since it corresponds to a constant
function $c$ in the classical case.

We can now immediately and unambiguously write down a noncommutative version
of the classical Ricci flow equation above:
\[
\frac{d}{dt}c(t)=-\triangle\log c\left(t\right)
\]
where $c\left(  t\right)  \in A$ denotes the metric at ``time'' $t$. We say
that this is unambiguous, since, given the Laplacian $\triangle$ on $A$, 
the equation does not depend on an arbitrary choice of
order of noncommuting elements of $A$. We interpret the time derivative in the
equation in terms of any norm on $A$, since these are equivalent as already
mentioned. In particular we can therefore view it as component-wise
differentiation, but equivalently we can view the derivative as being defined
in terms of the operator norm (i.e. the C*-algebraic norm) or the
Hilbert-Schmidt norm on $A$. For any strictly positive $a\in A$, $\log a$ is
defined by the usual functional calculus, namely transform unitarily to an orthonormal
basis in which $a$ is diagonal, take $\log$ of the diagonal elements, and then
transform back to the original orthonormal basis.

As an analogue of the cigar metric mentioned above, we have
\[
c=\left(M+x^{2}+y^{2}\right)^{-1}
\]
as an example of a noncommutative metric of the form above, where $M>0$ is a real number.
(Note that by $M+x^{2}+y^{2}$ we mean $M1+x^{2}+y^{2}$, but here and later on
we often drop the identity matrix $1$ when we add a scalar multiple of $1$ to
another matrix.) This is well defined because $M+x^{2}+y^{2}>0$ implying that
$M+x^{2}+y^{2}$ is indeed invertible, since $M>0$ while $x^{2},y^{2}\geq0$ due
to the fact that $x$ and $y$ are Hermitian. This is not a soliton for Ricci
flow though, since as will be seen in the next section all metrics flow to a
constant metric as is the case for metrics on the classical torus.

\section{Existence and convergence}

In this section we study the mathematical properties of the noncommutative
Ricci flow equation introduced in the previous section. Note that this
equation is in fact a system of first order ordinary differential equations,
and the theory of such systems will therefore play a central role in this
section (see for example \cite{Ch} and \cite{MM}). In particular we show the
existence of the flow on any interval $[t_{0},\infty)$, and show that in the
$t\rightarrow\infty$ limit the Ricci flow takes any initial metric to a flat
metric. We continue with the notation introduced in the previous section.

The following inequality will shortly be used in the proof of our main result.
In it we use the exponential $e^{a}$ of a Hermitian element $a$ of $A$, which
can equivalently be defined by the functional calculus as for $\log$ above, or
by a power series, or by the analytic functional calculus.

\begin{proposition}
For any Hermitian $a\in A$ we have
\[
\tau\left(e^{a}\triangle a\right)\geq 0
\]
with equality if and only if $a\in\mathbb{C}1$, i.e. if and only if $a$ is a
scalar multiple of the identity matrix.
\end{proposition}

\begin{proof}
Let $\delta$ denote either of the two derivations $\delta_{1}$ or $\delta_{2}$. 
By Proposition 2.1(b)
\[
\tau\left(  e^{a}\delta^{2}a\right)=-\tau\left(  (\delta a)\delta e^{a}\right)
=-e^{\lambda}\tau\left( (\delta\left(a-\lambda\right) )\delta e^{a-\lambda}\right)
\]
where $\lambda$ is the smallest eigenvalue of $a$, so $h:=a-\lambda\geq 0$. We
consider
\[
\tau\left( (\delta h)\delta e^{h}\right)
=\sum_{j=0}^{\infty}\frac{1}{j!}\tau\left( (\delta h)\delta h^{j}\right)
\]
and show that each term in this series is negative (or zero). By the product
rule, $\tau\left( (\delta h)\delta h^{j}\right)$ is a sum of
terms of the form
\begin{eqnarray*}
\tau\left( (\delta h)h^{p}(\delta h)h^{q}\right)   
& = & \tau\left(h^{q/2}(\delta h)h^{p/2}h^{p/2}(\delta h)h^{q/2}\right)\\
& = & -\tau\left( (h^{p/2}(\delta h)h^{q/2}) ^{\ast}h^{p/2}(\delta h)h^{q/2}\right)\\
& \leq & 0
\end{eqnarray*}
by Proposition 2.1(c) and the fact that $h^{p/2}$ and $h^{q/2}$ are Hermitian
since $h\geq0$, where $p,q\in\left\{  0,1,2,...\right\}  $ with $p+1+q=j$.
Thus indeed $\tau\left( (\delta h)\delta h^{j}\right)\leq0$
and therefore $\tau\left( (\delta h)\delta e^{h}\right)\leq0$, 
hence $\tau\left(e^{a}\delta^{2}a\right)\geq0$, so $\tau\left(e^{a}\triangle a\right)\geq 0$.

Now suppose $\tau\left(e^{a}\triangle a\right)=0$. Since 
$\tau\left(e^{a}\delta_{\mu}^{2}a\right)\geq0$, it follows that 
$\tau\left(e^{a}\delta_{\mu}^{2}a\right)=0$ for $\mu=1,2$. Again in terms of the
notation above, it follows from $\tau\left( (\delta h)\delta h^{j}\right)\leq 0$ 
that $\tau\left( (\delta h)\delta h^{j}\right)=0$, in particular 
$0=\tau\left( (\delta h)\delta h\right)=-\left\langle\delta h,\delta h\right\rangle$. 
So $\delta_{\mu}a=\delta_{\mu}\left(h+\lambda\right)=0$, hence $a\in\mathbb{C}1$ 
by Proposition 2.1(a). The converse is trivial.
\end{proof}

The central result of this paper is the following (note that in its proof 
finite-dimensionality plays an essential role):

\begin{theorem}
Let $c_{0}$ be any initial noncommutative metric at the initial time 
$t_{0}\in\mathbb{R}$. Then the noncommutative Ricci flow equation has a unique
solution $c$ on the interval $[t_{0},\infty)$. Furthermore, the flow preserves
the trace, i.e. $\tau\left(c\left(t\right)\right)=\tau\left(c_{0}\right)$ 
for all $t\geq t_{0}$, and
\[
\lim_{t\rightarrow\infty}c\left(t\right)=c_{\infty}1
\]
in any of the equivalent norms on $A$, where
\[
c_{\infty}:=\frac{1}{n}\tau\left(c_{0}\right)
\]
with $n$ the dimension as before. Lastly, the determinant $\det c\left(t\right)$
is strictly increasing in $t$, unless $c_{0}=c_{\infty}$ in which case
$c\left(t\right)=c_{0}$ for all $t\in\mathbb{R}$.
\end{theorem}

\begin{proof}
The set $A_{sa}$ of all Hermitian elements of $A$ is a real vector space
(not that its elements have real entries, but it is a vector space
with respect to real scalar multiplication). If $c\left(t\right)>0$,
then $\log c\left(t\right)$ is Hermitian, hence so is $\triangle\log c\left(t\right)$ 
because of Proposition 2.1(c). So the Ricci flow equation is expressed wholly in 
terms of the space $A_{sa}$. In particular $\frac{d}{dt}c(t)$ is Hermitian, hence
$c\left(t\right)$ at least remains in $A_{sa}$ under the flow, so the eigenvalues of
$c\left(t\right)$ remain real under the Ricci flow. However, due to 
the $\log$, the equation is only defined on the set of strictly
positive matrices. (In conventional ordinary differential equation terms, 
the equation is formulated for some domain in $\mathbb{R}^{n^2}$, the latter 
representing $A_{sa}$.) Since the set of strictly positive matrices is open 
in $A_{sa}$, we have solutions on open intervals which are short enough, 
and in order to show that the Ricci flow exists for all $t\geq t_{0}$, 
we need to show that $c\left(t\right)>0$ remains true under the flow, 
and that the solution does not blow up in finite time. The solution will necessarily be 
unique, since it is a system of first order ordinary differential equations.

First note that by Proposition 2.1(g),
\[
\frac{d}{dt}\tau\left(c\right)=\tau\left(\frac{dc}{dt}\right)
=-\tau\left(\triangle\log c\right)=0
\]
on any interval on which the solution exists, i.e. the trace is preserved
under Ricci flow. By Proposition 3.1 and the Ricci flow equation, we have
\[
0\leq\tau\left(e^{-\log c}\triangle (-\log c)\right)
=\tau\left(c^{-1}\frac{dc}{dt}\right)=\frac{d}{dt}\log(\det c)
\]
with equality to $0$ if and only if $c\left(t\right)\in\mathbb{C}1$, and
where the last equality is a standard identity from matrix analysis
(see for example \cite[Section 6.5]{HJ} where the proof is outlined using
properties of $\det c$). It follows that $\det c$ is strictly increasing in
$t$, or constant in the case $c\left(t\right)\in\mathbb{C}1$, therefore
$\det c$ remains larger than some strictly positive real number under the
Ricci flow. However, since Ricci flow preserves the trace, none of the
eigenvalues of $c$ can become bigger than $nc_{\infty}$ under the flow while
they are all positive, so should one of the eigenvalues go to zero under the
flow, it would follow that $\det c\left(t\right)$ would go to zero, which
is a contradiction. Since none of the eigenvalues go to zero, this then also
proves that all the eigenvalues remain strictly positive and bounded by
$nc_{\infty}$, which means that $c\left(t\right)>0$ remains true under the
flow and the solution $c$ does not blow up. Therefore we have shown that the
flow exists for all $t\in\lbrack t_{0},\infty)$.

Also note that $c_{0}\neq c_{\infty}$ along with the uniqueness of the
solution, implies that $c\left(t\right)$ is not in $\mathbb{C}1$ for any
$t$. For if $c\left(t\right)$ was in $\mathbb{C}1$ for some $t$, it would
be there for all $t\in\lbrack t_{0},\infty)$, since it is clearly a fixed
point of the Ricci flow.

All that remains is to show that $c\left(t\right)$ converges to
$c_{\infty}1$. It is going to be convenient to work in terms of the
Hilbert-Schmidt norm $\left\|\cdot\right\|_{2}$ given by 
$\left\|a\right\|_{2}:=\left\langle a,a\right\rangle ^{1/2}=\tau\left(a^{\ast}a\right)^{1/2}$
for all $a\in A$. Consider the function $v$ on $A$ defined
by
\[
v\left(a\right):=\left\|a-c_{\infty}\right\|_{2}^{2}
\]
and note that since Ricci flow preserves the trace $\tau$, we have
\[
\frac{d}{dt}v\left(c\right)=\tau\left(\frac{d}{dt}(c-c_{\infty})^{2}\right)
=2\tau\left(  c\frac{dc}{dt}\right)=-2\tau\left(c\triangle\log c\right)\leq 0
\]
by Proposition 3.1. In other words, $v(c(t))$ is decreasing in $t$
and bounded from below by $0$, therefore
\[
L:=\lim_{t\rightarrow\infty}v\left(c\left(t\right)\right)
\]
exists. We show that $L=0$, since that will imply that $c\left(t\right)$
converges to $c_{\infty}1$. Set
\[
V:=\left\{a\in A_{sa}:\frac{1}{n}\tau\left(e^{a}\right)
=c_{\infty}\quad\mbox{and}\quad L\leq v\left(e^{a}\right)\leq L+1\right\}
\]
which is norm closed, since $A_{sa}$ is. But $V$ is also bounded and by 
finite-dimensionality of $A$ it follows that $V$ is compact. By definition of $L$ and
the conservation of $\tau\left(c\left(t\right)\right)$ under Ricci
flow, there exists a $t_{1}\geq t_{0}$ such that $\log c\left(t\right)\in V$ 
for all $t>t_{1}$. Also note that $\frac{d}{dt}v\left(c\left(t\right)\right)$ 
comes arbitrarily close to $0$ for $t>t_{1}$, since suppose that
$\frac{d}{dt}v\left(c\left(t\right)\right)\leq-\varepsilon$ for all $t>t_{1}$ 
for some $\varepsilon>0$, then $v\left(c\left(t\right)\right)\leq v_{0}-\varepsilon t$ 
for all $t>t_{1}$ for some $v_{0}\in\mathbb{R}$, contradicting the fact that 
$v\left(c\left(t\right)\right)\geq 0$ for all $t>t_{0}$ by the definition of $v$.
Furthermore, the function
\[
f:A_{sa}\rightarrow\mathbb{C}:a\mapsto-2\tau\left(e^{a}\triangle a\right)
\]
is continuous, so $f\left(V\right)$ is compact and therefore closed. But
$f\left(\log c\right)=\frac{d}{dt}v\left(c\right)$ comes arbitrarily
close to $0$ for $t>t_{1}$ therefore $0\in f\left(V\right)$. In other
words there is an $a\in A_{sa}$ such that $\frac{1}{n}\tau\left(e^{a}\right)=c_{\infty}$, 
$v\left(  e^{a}\right)  \geq L\geq 0$ and $\tau\left(e^{a}\triangle a\right)=0$.
However, the last condition implies that $a\in\mathbb{C}1$ by Proposition 3.1, 
so $e^{a}=c_{\infty}1$ by the first condition, which means that 
$v\left(e^{a}\right)=0$. We conclude that $L=0$.
\end{proof}

Note that $c\left(t\right)=c_{\infty}1$ is the unique fixed point of the
Ricci flow which has trace $nc_{\infty}$, since if $dc/dt=0$, then
$\triangle\log c=0$, so by Proposition 2.1(f) $c\left(  t\right)
\in\mathbb{C}1$ from which it follows that $c\left(t\right)=c_{\infty}1$.
The theorem above says that the metric flows to this fixed point. Furthermore,
the preservation of the trace corresponds the the preservation of the total
area in the classical case. Also keep in mind that the solution referred to in
the theorem is continuously differentiable by the general theory of systems of
first order ordinary differential equations, although we can also see it from
the fact that $\triangle\log c$ is differentiable, since $c$ is (see for
example \cite[Section 6.6]{HJ}), which means that $d^{2}c/dt^{2}$ exists by
the noncommutative Ricci flow equation.

\section{Density matrices and entropy}

In Theorem 3.2 we found that the Ricci flow exists and converges to a flat 
metric. It is now natural to study further qualitative features of the flow. 
In particular we study a monotonicity property of the flow in terms of an 
entropy. In the classical case monotonicity properties are often useful in a 
qualitative analysis of geometric flows. See \cite[Section 8]{H2} for a 
case in point. One can therefore expect monotonicity properties to be useful in 
the noncommutative case as well. In fact, we already saw a simple instance of 
this as part of the proof of Theorem 3.2 where the monotonicity of $\det c$, or 
equivalently of $\log(\det c)=\tau(\log c)$, helped us to show the existence of 
the Ricci flow on $[t_0,\infty )$. Here we consider a particularly simple 
entropy similar in form to the function $\tau(\log c)$ just mentioned, which 
can likewise be defined directly in terms of the metric and appears to fit very 
naturally into the theory we have set up so far, but is interestingly enough 
also well known from usual quantum mechanics.

Note that if in the noncommutative Ricci flow equation we set $t^{\prime}=\kappa t$ 
and $\rho\left(t^{\prime}\right)=\kappa c\left(t^{\prime}/\kappa\right)$ 
for any constant $\kappa>0$, then
\[
\frac{d}{dt}\rho\left(t\right)=-\triangle\log\rho\left(t\right)
\]
(where we have written $t$ instead of $t^{\prime}$) so we have simply scaled
the Ricci flow. In particular, by the conservation of trace given by Theorem
3.2, we can choose $\kappa$ such that $\rho\left(t\right)$ is a strictly
positive density matrix, meaning $\rho\left(t\right)>0$ and 
$\tau\left(\rho\left(t\right)\right)=1$, for all $t\geq t_{0}$. By the 
previous section we then know that under the Ricci flow the
density matrix flows to the unique fixed point of the flow, which in this case
is the density matrix of maximum von Neumann entropy, namely $\frac{1}{n}1$.
But we can say more:

\begin{theorem}
Consider the noncommutative Ricci flow of a strictly positive density matrix
$\rho\left(t\right)$, i.e. $c=\rho$ in Theorem 3.2. Let
\[
S:=-\tau\left(\rho\log\rho\right)
\]
be the von Neumann entropy of $\rho$. Then either $\rho\left(t\right)$ is the
density matrix of maximum entropy (the fixed point of the Ricci flow), or $S$
is a strictly increasing function of $t$ converging to the maximum entropy as
$t\rightarrow\infty$.
\end{theorem}

\begin{proof}
The first fact we need is that for any Hermitian $a\in A$ we have
\[
\tau\left(a\triangle a\right)\geq 0
\]
with equality if and only if $a\in\mathbb{C}1$. The inequality follows from
Proposition 2.1(d). If $\tau\left(a\triangle a\right)=0$, it also follows
from Proposition 2.1(d) that $\left\langle a,\delta_{\mu}^{2}a\right\rangle=0$, 
so $\delta_{\mu}a=0$ by Proposition 2.1(e), hence $a\in\mathbb{C}1$ by
Proposition 2.1(a).

The second fact we prove is the following matrix analysis identity:
\[
\tau\left(e^{l}\frac{dl}{dt}\right)=\frac{d}{dt}\tau\left(e^{l}\right)
\]
where we have set $l:=\log\rho$. To see this, begin by noticing that since
$d\rho/dt$ exists and is continuous as mentioned in Section 3, the
same is true for $dl/dt$ (see for example \cite[Section 6.6]{HJ}).
Regarding the series expansion
\[
\tau\left(e^{l}\right)=\sum_{j=0}^{\infty}\frac{1}{j!}\tau\left(l^{j}\right)
\]
we have
\[
\frac{d}{dt}\tau\left(l^{j}\right)
=\tau\left(\frac{dl}{dt}l^{j-1}+l\frac{dl}{dt}l^{j-2}+...+l^{j-1}\frac{dl}{dt}\right)  
=j\tau\left(l^{j-1}\frac{dl}{dt}\right)
\]
so on any time interval $\left[t_{1},t_{2}\right]$ we have in terms of the
operator norm $\left\|\cdot\right\|$ on $A$ that for every $t\in\left[t_{1},t_{2}\right]$,
\[
\left\|\frac{d}{dt}\tau\left(l\left(t\right)^{j}\right)\right\|\leq jnM^{j-1}N
\]
where $M$ and $N$ are respectively the maxima of $t\mapsto\left\|l\left(t\right)\right\|$ 
and $t\mapsto\left\|\frac{d}{dt}l\left(t\right)\right\|$ on $\left[t_{1},t_{2}\right]$ 
(these maxima exist, since $l$ and $dl/dt$ are continuous). Since
\[
\sum_{j=0}^{\infty}\frac{1}{j!}jnM^{j-1}N
\]
is convergent, it follows from the Weierstrass test that
\[
\sum_{j=0}^{\infty}\frac{1}{j!}\frac{d}{dt}\tau\left(l^{j}\right)
\]
converges uniformly on $\left[t_{1},t_{2}\right]$, hence we can
differentiate the series for $\tau\left(e^{l}\right)$ termwise:
\[
\frac{d}{dt}\tau\left(e^{l}\right)
=\sum_{j=0}^{\infty}\frac{1}{j!}\frac{d}{dt}\tau\left(  l^{j}\right)
=\sum_{j=1}^{\infty}\frac{1}{\left(j-1\right)!}\tau\left(l^{j-1}\frac{dl}{dt}\right)
=\tau\left(e^{l}\frac{dl}{dt}\right)
\]
as required. (Note that this works for general $c$ instead of $\rho$ as well.)

However, as $\rho$ is a solution of the Ricci flow, it follows from this
identity that
\[
\tau\left(\rho\frac{d}{dt}\log\rho\right)=\frac{d}{dt}\tau\left(\rho\right)=0
\]
by Theorem 3.2. Now, if $S$ is not at its maximum and therefore $\rho$ and
thus $l$ are not in $\mathbb{C}1$, it follows by again using the Ricci flow
equation that
\[
\frac{d}{dt}S
=-\tau\left(\frac{d\rho}{dt}\log\rho\right)-\tau\left(\rho\frac{d}{dt}\log\rho\right)
=\tau\left(l\triangle l\right)
>0
\]
from the inequality above. That is to say, $S$ is strictly increasing in $t$.
That $S$ converges to the maximum entropy follows from Theorem 3.2, as already 
mentioned.
\end{proof}

Note that in the classical case various entropies related to Ricci flow have 
been studied extensively (see for example \cite[Section 7]{H2}, \cite{Cho} 
and \cite{P1}) but they are quite different from the entropy we studied here, 
even aside from the noncommutative setting, as they are not defined directly 
in terms of the metric, but typically rather in terms of the scalar curvature.

\section{Scalar curvature}

If we define the classical scalar curvature $R$ in terms of the Ricci tensor
$R_{\mu\nu}$ by $R:=R_{\mu}^{\mu}$ as usual, then in the case of the classical
metric $g_{\mu\nu}=c\delta_{\mu\nu}$, we obtain
\[
R=-\frac{1}{c}\partial_{\mu}\partial_{\mu}\log c
\]
which leads to ambiguity when adapted directly to the noncommutative case,
i.e. in general $c^{-1}$ and $\triangle\log c$ will not commute in
$c^{-1}\triangle\log c$. However, one can naturally obtain $R$ from the Ricci
flow without any such ambiguity in the noncommutative case. Indeed, in the
classical case it is easily verified that given a solution $c$ to the Ricci
flow,
\[
R\left(t\right)=-\frac{\partial}{\partial t}\log c\left(t\right)
\]
is the scalar curvature corresponding to the metric $c\left(t\right)$.

The corresponding formula can be used without any ambiguity in the
noncommutative case as the \emph{definition} of the scalar curvature at any
time $t$ in the noncommutative Ricci flow:
\[
R\left(t\right):=-\frac{d}{dt}\log c\left(t\right)
\]
This formula corresponds to \cite[eq. (121)]{CM}, although there it is not a
definition, but follows from an alternative definition of scalar curvature. In
particular, given a noncommutative metric $c_{0}$, simply consider the Ricci
flow $c$ starting there, at time $t=0$, and the scalar curvature of $c_{0}$ is
then defined to be
\[
R_{0}:=-\frac{d}{dt}\log c\left(t\right)|_{t=0}
\]
Since $\log c\left(  t\right)  $ is a Hermitian matrix, so is $R_{0}$, so it
corresponds to a real-valued function on a classical surface as it should, in
contrast to $c^{-1}\triangle\log c$ or $\left(\triangle\log c\right)c^{-1}$ 
which are not in general Hermitian.

The noncommutative scalar curvature will not be easy to evaluate analytically
in general, although numerically it would be possible. Nevertheless we can
still study some of properties of the noncommutative scalar curvature defined
in this way, which is what we now turn to.

The basic property is that the scalar curvature has zero average, as it does
on the classical torus. Keep in mind that in the classical case the average
follows from an integral over the surface, where we need to include the metric
as a factor to have the correct measure, i.e. the integral of a function $f$
over the surface is given by $\int fcdx^{1}dx^{2}$. Similarly we include the
metric $c$ in the noncommutative case as well, exactly as is done in
\cite{CT}. We formulate this along with some related elementary properties:

\begin{proposition}
If $R_{0}$ is the scalar curvature of the noncommutative metric $c_{0}$,
then
\[
\tau\left(c_{0}R_{0}\right)=0
\]
In particular, if $R_{0}$ is constant, i.e. $R_{0}\in\mathbb{C}1$, then
$R_{0}=0$, which in turn holds if and only if $c_{0}$ is constant, i.e. the
metric is flat.
\end{proposition}

\begin{proof}
Using the identity derived in the proof of Theorem 4.1, we have
\[
\tau\left(c_{0}R_{0}\right) 
= -\tau\left. \left(c\frac{d}{dt}\log c\right)\right|_{t=0}
= -\left. \frac{d}{dt}\tau\left(c\right)\right|_{t=0}
= 0
\]
by Theorem 3.2.

If $R_{0}$ is constant, say $R_{0}=r1$ for some real number $r$, it follows
that $r=\tau\left(c_{0}R_{0}\right)/\tau\left(c_{0}\right)=0$, so
$R_{0}=0$.

If $c_{0}$ is constant, i.e. $c_{0}\in\mathbb{C}1$, then it is a fixed point
of the Ricci flow, so $R_{0}=0$. Conversely, if $R_{0}=0$, then using the
Ricci flow equation as well as the matrix analysis identity appearing in
the proof of Theorem 3.2, we have
\begin{eqnarray*}
\tau\left(e^{-\log c_{0}}\triangle (-\log c_{0})\right)   
& = & \tau\left. \left(c^{-1}\frac{dc}{dt}\right)\right|_{t=0}
= \frac{d}{dt}\log\left(\det c\right)|_{t=0} \\
& = & \left. \frac{d}{dt}\tau\left(  \log c\right)\right|_{t=0}
= -\tau\left(R_{0}\right)
= 0
\end{eqnarray*}
from which it follows that $c_{0}\in\mathbb{C}1$ by Proposition 3.1.
\end{proof}

\section{Further work}

A number of points have not been addressed in this paper. One is the existence
and behaviour of the noncommutative Ricci flow on $(-\infty,t_{0}]$. Another
is whether the convergence of the noncommutative Ricci flow is exponential
(see \cite[Section 5]{H2} for the classical case). One could also consider
generalizing the setting. For example in \cite{DS} and \cite{Ro} more general
setups for metrics on the noncommutative torus are considered. And of course
one could attempt to study Ricci flow and the resulting scalar curvature on
other matrix geometries, like the fuzzy sphere \cite{M}, or even more general
frameworks like noncommutative Riemann surfaces \cite{A2}, \cite{A3}. It would
also be possible to study the noncommutative Ricci flow equation numerically.

The evolution of $R$ under Ricci flow appears to be another worthwhile problem
to explore. For example to find lower and upper bounds for the smallest and
largest eigenvalue of $R\left(t\right)$ respectively, in analogy to the
classical case \cite{H2} (also see for example \cite[Chapter 5]{CK}). This
would also open up the possibility to study entropy defined in terms of the
scalar curvature as in for example \cite{H2} and \cite{Cho}.

A study of the connection between the noncommutative Ricci flow as studied in
this paper and the other approaches mentioned in the Introduction might be
insightful. This would somehow have to involve a large $n$ ``limit'' of the
approach in this paper, in which $m/n$ in $q=e^{2\pi im/n}$ would have to
converge to a specified value $\theta$. Possibly ideas from \cite{PV},
\cite{LLS} and \cite{AG}, or \cite{R} and \cite{L} would be relevant here.

More speculatively, the fact that $c$ can be normalized to a density matrix by
scaling the ``time'' parameter is suggestive that there may be an
interpretation or use of the noncommutative Ricci flow in quantum mechanics. 
However, since the Ricci flow equation is not linear, it is not the physical 
time-evolution of a quantum system. Nevertheless it may be interesting to 
explore whether it has applications in for example quantum information. More 
broadly it may be fruitful to study if the geometric view of a density matrix 
as a metric in noncommutative geometry has any value in quantum mechanics and 
quantum information and if the Ricci scalar curvature has a quantum mechanical 
interpretation.

\textbf{Acknowledgement.} This research was supported by the National Research 
Foundation of South Africa. I thank the referees for suggestions to improve the
exposition.

\end{document}